\documentclass[english]{article}
\usepackage[T1]{fontenc}
\usepackage[latin1]{inputenc}
\usepackage{geometry}
\usepackage{float}
\usepackage{mathrsfs}
\usepackage{amsmath}
\usepackage{amssymb}
\usepackage{graphicx}

\makeatletter

\floatstyle{ruled}
\newfloat{algorithm}{tbp}{loa}
\providecommand{\algorithmname}{Algorithm}
\floatname{algorithm}{\protect\algorithmname}

\usepackage{amsthm}

\usepackage{mathrsfs}

\usepackage{amsfonts}

\usepackage{epsfig}

\usepackage{bm}

\usepackage{mathrsfs}

\usepackage{enumerate}

\@ifundefined{definecolor}{\@ifundefined{definecolor}
 {\@ifundefined{definecolor}
 {\usepackage{color}}{}
}{}
}{}

\newtheorem{theorem}{Theorem}[section]
\newtheorem{lem}{Lemma}[section]
\newtheorem{rem}{Remark}[section]

\newcounter{hypA}
\newenvironment{hypA}{\refstepcounter{hypA}\begin{itemize}
  \item[({\bf A\arabic{hypA}})]}{\end{itemize}}

\usepackage{babel}\date{}

\usepackage{babel}

\makeatother

\usepackage{babel}
\begin{document}

\begin{center}

{\Large \textbf{Marginal Likelihood Computation for Hidden Markov Models via Generalized Two-Filter Smoothing}}

\vspace{0.5cm}

BY ADAM PERSING$^{1}$ \& AJAY JASRA$^{2}$ 

{\footnotesize $^{1}$Department of Mathematics,
Imperial College London, London, SW7 2AZ, UK.}\\
{\footnotesize E-Mail:\,}\texttt{\emph{\footnotesize a.persing11@ic.ac.uk}}\\
{\footnotesize $^{2}$Department of Statistics \& Applied Probability,
National University of Singapore, Singapore, 117546, SG.}\\
{\footnotesize E-Mail:\,}\texttt{\emph{\footnotesize staja@nus.edu.sg}}
\end{center}

\begin{abstract}
In this note we introduce an estimate for the marginal likelihood associated to hidden Markov models (HMMs)
using sequential Monte Carlo (SMC) approximations of the generalized two-filter smoothing decomposition \cite{Briers_2010}.
This estimate is shown to be unbiased and a central limit theorem (CLT) is established. 
This latter CLT also allows one to prove a CLT associated to estimates of expectations w.r.t.~a marginal  of the joint smoothing distribution; these form some
of the first theoretical results associated to the SMC approximation of the generalized two-filter smoothing decomposition. 
The new estimate and its application is investigated from a numerical perspective.
\end{abstract}

\section{Introduction}

Hidden Markov models provide a flexible description of a wide variety of real-life phenomena; see \cite{cappe}. An HMM\ is a pair of discrete-time stochastic
processes, $\left\{  X_{n}\right\}  _{n\mathbb{\geq}0}$ and $\left\{
Y_{n}\right\}  _{n\geq1}$, where $X_{n}\in\mathbb{R}^{d_{x}}$ is an unobserved
process and $y_{n}\in\mathbb{R}^{d_{y}}$ is observed. The hidden process
$\left\{  X_{n}\right\}_{n\geq 0}$ is a  Markov chain with initial density $\delta_{x_{0}}  $ at time $0$ and transition density $f_{\theta}\left(x_{n}|x_{n-1}\right)
$, with $\theta\in\Theta\subseteq\mathbb{R}^{d_{\theta}}$ i.e.
$\mathbb{P}_{\theta}(X_{0}\in A)=\delta_{x_{0}}(A)$ and
$\mathbb{P}_{\theta}(X_{n}\in A|X_{n-1}=x_{n-1})=\int_{A}f_{\theta}(x_{n}|x_{n-1}%
)dx_{n}~n\geq1~$
where $\mathbb{P}_{\theta}$ denotes probability, $A\subseteq\mathbb{R}^{d_{x}}$, $\delta_{x_{0}}$ is the
Dirac measure with mass at $x_0$, and $dx_{n}$ an assumed dominating measure.
In addition, the observations
$\left\{  Y_{n}\right\}  _{n\geq1}$\ conditioned upon $\left\{  X_{n}\right\}
_{n\mathbb{\geq}0}$ are statistically independent and have marginal density
$g_{\theta}\left(y_{n}|x_{n}\right)  $, i.e.%
$\mathbb{P}_{\theta}(Y_{n}\in B|\{X_{k}\}_{k\geq 0}=\{x_{k}\}_{k\geq1})=\int_{B}%
g_{\theta}(y_{n}|x_{n})dy_{n}~n\geq1$
with $B\subseteq\mathbb{R}^{d_{y}}$ and $dy_{n}$ the dominating measure. The HMM described above is often referred to in the literature as a state-space model. 
Here $\theta$ is a static parameter, which is fixed throughout and we shall only be concerned with scenario that one observes a batch data set $y_{1:T}:=(y_1,\dots,y_T)$.
The joint density of the observations $p_{\theta}(y_{1:T})$ is termed the marginal likelihood. For most models of practical interest, this quantity cannot be evaluated exactly.
A popular collection of approximation techniques for HMMs, which can estimate the marginal likelihood are SMC methods.


SMC techniques simulate a collection of $N$ samples in parallel, sequentially in time and combine importance sampling and resampling to approximate
a sequence of probability distributions of increasing state-space known up-to an additve constant; see \cite{Doucet_2001} for an introduction. These techniques provide a natural estimate of the marginal likelihood of HMMs
(as well as for normalizing constants of Feynman-Kac representations; see \cite{delmoral}).  The estimate
is quite well understood and is known to be unbiased \cite{delmoral}
and the relative variance is known to increase linearly with $T$ \cite{cerou1,whiteley1}. However, the standard SMC estimate
is not the only alternative one can consider. A relatively recent procedure designed for smoothing, is based upon the \emph{generalized} two-filter decomposition (see e.g.~\cite{bresler} for the two-filter smoothing decomposition). Roughly, the idea is to run two independent SMC algorithms, one forwards
(as before) and one backwards (which approximates a collection of appropriately defined target distributions) and for them to `meet' at some point. Using this procedure, one can yield more efficient schemes for smoothing, relative
to standard SMC procedures. In the following note we:
\begin{enumerate}
\item{Introduce a new estimate, costing $\mathcal{O}(N)$, of the marginal likelihood using the generalized two-filter smoothing decomposition.}
\item{Establish that this estimate is unbiased and prove a CLT, under some assumptions.}
\item{Numerically investigate the estimate.}
\end{enumerate}
It is remarked that via 2.~we can also establish a CLT for an estimate of expectations w.r.t.~a marginal of the joint smoothing distribution.

This note is in two halves; the first focuses on the idea from a methodological perspective. The second is the proof of our results in point 2. The note is structured as follows:
in Section \ref{sec:two_filter} we discuss the estimate and our main result. In Section \ref{sec:pmcmc}  some simulations
investigating the new estimate are given; in particular, some comparisons to the forward filtering backward simulation (FFBSi) algorithm in \cite{douc}.
The proofs of our results are housed in the appendix.

\section{SMC and Generalized Two-Filter Smoothing}\label{sec:two_filter}


\subsection{SMC Algorithm}

We consider the joint smoothing distribution, with $\theta$ fixed:
\begin{equation}
\pi_{\theta}(x_{1:T}|y_{1:T}) = 
\frac{\prod_{n=1}^T g_{\theta}(y_n|x_n)f_{\theta}(x_n|x_{n-1})}{\int_{\mathbb{R}^{Td_x}}\prod_{n=1}^T g_{\theta}(y_n|x_n)f_{\theta}(x_n|x_{n-1}) dx_{1:T}}
\label{eq:smoother}
\end{equation}
the denominator is denoted $p_{\theta}(y_{1:T})$; this is the marginal likelihood.
We remark that throughout, the transition and observation densities can be time-inhomogeneous, but we omit this from our notation. One can construct an SMC algorithm to sample
sequentially from $\pi_{\theta}(x_{1}|y_{1}),\dots,\pi_{\theta}(x_{1:T}|y_{1:T})$. The idea is to use a collection of particles, simulated in parallel, which are written
$(\overrightarrow{X}_{1:n}^i)_{i\in\{1,\dots,N\}}$ to denote samples forward in time, the reason for the notation will become apparent below.  We will sometimes denote the index of a particle at time $n$ by $\overrightarrow{a}_n^i$, and we adopt the notation $\overrightarrow{x}_n^{\overrightarrow{a}_n^i} = \overrightarrow{x}_n^{a(i)}$. 

\begin{itemize}
\item{Step 1: For $i\in\{1,\dots,N\}$ sample $\overrightarrow{X}_1^i\sim q_{1,\theta}(\cdot)$ and compute the un-normalized weight:
$$
\overrightarrow{W}_1^i = \frac{g_{\theta}(y_1|\overrightarrow{x}_1^i)f_{\theta}(\overrightarrow{x}_1^i|x_0)}{q_{1,\theta}(\overrightarrow{x}_1^i)}.
$$
For $i\in\{1,\dots,N\}$ sample $\overrightarrow{a}_1^i\in\{1,\dots,N\}$ from a discrete distribution on $\{1,\dots,N\}$ with $jth$ probability
$
\overrightarrow{w}_1^j = \overrightarrow{W}_1^j/\sum_{l=1}^N \overrightarrow{W}_1^l
$
these represent the resampled particles. Set $n=2$.}
\item{Step 2: If $n=T+1$ stop. Otherwise, for $i\in\{1,\dots,N\}$ sample $\overrightarrow{X}_n^i|\overrightarrow{x}_{n-1}^{a(i)}\sim q_{n,\theta}(\cdot|\overrightarrow{x}_{n-1}^{a(i)})$ and compute the un-normalized weight:
$$
\overrightarrow{W}_n^i = \frac{g_{\theta}(y_n|\overrightarrow{x}_n^i)f_{\theta}(\overrightarrow{x}_n^i|\overrightarrow{x}_{n-1}^{a(i)})}
{q_{n,\theta}(\overrightarrow{x}_n^i|\overrightarrow{x}_{n-1}^{a(i)})}.
$$
For $i\in\{1,\dots,N\}$ sample $\overrightarrow{a}_n^i\in\{1,\dots,N\}$ from a discrete distribution on $\{1,\dots,N\}$ with $jth$ probability
$
\overrightarrow{w}_n^j = \overrightarrow{W}_n^j/\sum_{l=1}^N \overrightarrow{W}_n^l.
$
Set $n=n+1$ and return to the start of step 2.}
\end{itemize}

The estimate of the marginal likelihood is:
\begin{equation}
p_{\theta}^N(y_{1:T}) = \prod_{n=1}^{T} \left( \frac{1}{N}\sum_{l=1}^N \overrightarrow{W}_n^l\right).
\label{eq:smc_nc}
\end{equation}

\subsection{Generalized Two-Filter Smoothing}

It is well-known that the above SMC algorithm does not approximate the joint smoothing distribution at all well. One technique which is known to assist the simulation procedure (at least empirically) is the SMC approximation of the generalized two-filter representation \cite{Briers_2010}.
The algorithm works by defining two filters. One works as the SMC algorithm above and moves forward in time. The other works backward in time, on a sequence of densities defined below. These two algorithms `meet' at some pre-specified time $t\in\{1,\dots,T\}$.
 
Define the following sequence of densities (we use the convention $\prod_{\emptyset}=1$):
$$
\widetilde{\pi}_{\theta}(x_{n:T}|y_{n:T}) \propto \xi_{n,\theta}(x_t)g_{\theta}(y_n|x_n)
\bigg[\prod_{n=t+1}^T g_{\theta}(y_n|x_n)f_{\theta}(x_n|x_{n-1})\bigg] \quad n\in\{t,\dots,T\}
$$
where, at this stage, $\xi_{n,\theta}$ are a sequence of (essentially) arbitrary density functions w.r.t.~$dx_n$.
In practice, the $\xi_{n,\theta}$ are critical to the efficiency of the algorithm and we return to this point in Section \ref{sec:pmcmc}. 
We write the normalizing constant as $\tilde{p}_{\theta}(y_{n:T})$.
One can use SMC to approximate this sequence of densities.  We will sometimes denote the index of a particle at time $n$ by $\overleftarrow{a}_n^i$, and we adopt the notation $\overleftarrow{x}_n^{\overleftarrow{a}_n^i} = \overleftarrow{x}_n^{a(i)}$.
\begin{itemize}
\item{Step 1: For $i\in\{1,\dots,N\}$ sample $\overleftarrow{X}_T^i\sim q_{T,\theta}(\cdot)$ and compute the un-normalized weight:
$$
\overleftarrow{W}_T^i = \frac{\xi_{T,\theta}(\overleftarrow{x}_T^i)g_{\theta}(y_T|\overleftarrow{x}_T^i)}
{q_{T,\theta}(\overleftarrow{x}_1^i)}.
$$
For $i\in\{1,\dots,N\}$ sample $\overleftarrow{a}_T^i\in\{1,\dots,N\}$ from a discrete distribution on $\{1,\dots,N\}$ with $jth$ probability
$
\overleftarrow{w}_T^j = \overleftarrow{W}_T^j/\sum_{l=1}^N \overleftarrow{W}_T^l
$
these represent the resampled particles. Set $n=T-1$.}
\item{Step 2: If $n=t-1$ stop. Otherwise
For $i\in\{1,\dots,N\}$ sample $\overleftarrow{X}_n^i|\overleftarrow{x}_{n+1}^{a(i)}\sim q_{n,\theta}(\cdot|\overleftarrow{x}_{n+1}^{a(i)})$ and compute the un-normalized weight:
$$
\overleftarrow{W}_n^i = \frac{
\xi_{n,\theta}(\overleftarrow{x}_n^i)
g_{\theta}(y_n|\overleftarrow{x}_n^i)f_{\theta}(\overleftarrow{x}_{n+1}^{a(i)}|\overleftarrow{x}_n^i)}
{\xi_{n+1,\theta}(\overleftarrow{x}_{n+1}^{a(i)})
q_{n,\theta}(\overrightarrow{x}_n^i|\overleftarrow{x}_{n+1}^{a(i)})}.
$$
For $i\in\{1,\dots,N\}$ sample $\overleftarrow{a}_n^i\in\{1,\dots,N\}$ from a discrete distribution on $\{1,\dots,N\}$ with $jth$ probability
$
\overleftarrow{w}_n^j = \overleftarrow{W}_n^j/\sum_{l=1}^N \overleftarrow{W}_n^l.
$
Set $n=n-1$ and return to the start of step 2.}
\end{itemize}

One can estimate the normalizing constant $\tilde{p}_{\theta}(y_{t+1:T})$ by using a similar expression to \eqref{eq:smc_nc};
this estimate is denoted $\tilde{p}_{\theta}^N(y_{t+1:T})$.

\subsection{Two Estimates of the Marginal Likelihood}\label{sec:est_marginal_like}

The objective here is to consider how one can use generalized two-filter smoothing to estimate the marginal likelihood. One can consider \cite[Proposition 3]{Briers_2010} which states that
$$
p_{\theta}(y_{1:T}) = \int\pi_{\theta}(x_{t-1},y_{1:t-1}) \widetilde{\pi}_{\theta}(x_t,y_{t:T}) \frac{f_{\theta}(x_t|x_{t-1})}{\xi_{t,\theta}(x_t)} dx_{t-1:t}.
$$
After some standard calculations, one has
\begin{eqnarray*}
p_{\theta}(y_{1:T}) & = & p_{\theta}(y_{1:t-2})  \widetilde{p}_{\theta}(y_{t+1:T})\int 
q_{t-1,\theta}(x_{t-1}|x_{t-2})q_{t,\theta}(x_t|x_{t+1})\overrightarrow{W}_{t-1}(x_{t-2:t-1})
\overleftarrow{W}_{t}(x_{t:t+1}) \times 
\\ & &
\pi_{\theta}(x_{t-2}|y_{1:t-2})\widetilde{\pi}_{\theta}(x_{t+1}|y_{t+1:T})
\frac{f(x_t|x_{t-1})}{\xi_{t,\theta}(x_t)}
dx_{t-2:t+1}
\end{eqnarray*}
whence an SMC estimate, one filter run up-to time $t-1$ forward and the other run backward to time $t$ with no resampling at the final time step only, of the marginal likelihood is
\begin{eqnarray*}
p_{\theta}^N(y_{1:T}) & = & p_{\theta}^N(y_{1:t-2})  \widetilde{p}_{\theta}^N(y_{t+1:T})
\frac{1}{N^2}\sum_{i=1}^N\sum_{j=1}^N
\overrightarrow{W}_{t-1}(\overrightarrow{x}^{i}_{t-2:t-1})\overleftarrow{W}_{t}(\overleftarrow{x}^{j}_{t:t+1}) \frac{f(\overleftarrow{x}_t^{j}|\overrightarrow{x}_{t-1}^i)}{\xi_{t,\theta}(\overleftarrow{x}_t^{j})}
\end{eqnarray*}
This estimate is perhaps slightly undesirable as it has a computational cost of $\mathcal{O}(N^2)$. 

An alternative approach is to use the slightly modified representation
$$
p_{\theta}(y_{1:T}) = p_{\theta}(y_{1:t-1})  \widetilde{p}_{\theta}(y_{t+1:T})\int 
\pi_{\theta}(x_{t-1}|y_{1:t-1})\widetilde{\pi}_{\theta}(x_{t+1}|y_{t+1:T}) 
\frac{f_{\theta}(x_t|x_{t-1})f_{\theta}(x_{t+1}|x_{t})}{\xi_{t+1,\theta}(x_{t+1})}g_{\theta}(y_t|x_t)
dx_{t-1:t+1}.
$$
Again, after some simple manipulations, one arrives at the formula
\begin{eqnarray*}
p_{\theta}(y_{1:T}) & = & p_{\theta}(y_{1:t-2})  \widetilde{p}_{\theta}(y_{t+2:T})\int 
\pi_{\theta}(x_{t-2}|y_{1:t-2})\widetilde{\pi}_{\theta}(x_{t+2}|y_{t+2:T}) 
\overrightarrow{W}_{t-1}(x_{t-2:t-1})\overleftarrow{W}_{t+1}(x_{t+1:t+2}) \times
\\
& & 
q_{t-1,\theta}(x_{t-1}|x_{t-2})q_{t+1,\theta}(x_{t+1}|x_{t})
\frac{f_{\theta}(x_t|x_{t-1})f_{\theta}(x_{t+1}|x_{t})}{\xi_{t+1,\theta}(x_{t+1})}g_{\theta}(y_t|x_t)
dx_{t-2:t+2}.
\end{eqnarray*}
Now, if one runs the two forward and backward SMC algorithms up-to times $t-1$ and $t+1$ respectively, not resampling at the very final
time steps, one has the approximation:
\begin{eqnarray*}
p_{\theta}^N(y_{1:T}) & = & p_{\theta}^N(y_{1:t-2})  \widetilde{p}_{\theta}^N(y_{t+2:T})
\frac{1}{N^2}\sum_{i=1}^N\sum_{j=1}^N
\overrightarrow{W}_{t-1}(\overrightarrow{x}^{i}_{t-2:t-1})\overleftarrow{W}_{t+1}(\overleftarrow{x}^j_{t+1:t+2}) \times
\\
& & 
\int\frac{f_{\theta}(x_t|\overrightarrow{x}_{t-1}^i)f_{\theta}(\overleftarrow{x}_{t+1}^j|x_{t})}{\xi_{t+1,\theta}(\overleftarrow{x}^j_{t+1})}g_{\theta}(y_t|x_t)
dx_{t}.
\end{eqnarray*}
This quantity can be approximated using the following procedure in \cite{Fearnhead_2010}.
Consider a conditional density $q_{t,\theta}(x_t|x_{t-1},x_{t+1})$ and two probabilities
$\overrightarrow{\beta}_{t-1}^i$, $\overleftarrow{\beta}_{t+1}^j$ , $i,j\in\{1,\dots,N\}$
$\sum_{i=1}^N\overrightarrow{\beta}_{t-1}^i=1$, $\sum_{j=1}^N\overleftarrow{\beta}_{t+1}^j=1$. Sample
$i(1),j(1),\dots,i(N),j(N)$ using the $\overrightarrow{\beta}_{t-1}^i$, $\overleftarrow{\beta}_{t+1}^j$ and then,
for each pair $i(l),j(l)$ sample $X_t^l|x_{t-1}^{i(l)},x_{t+1}^{j(l)}$ from the distribution induced by 
$q_{t,\theta}(\cdot|x_{t-1}^{i(l)},x_{t+1}^{j(l)})$, which leads to the estimate, which only costs $\mathcal{O}(N)$:
\begin{eqnarray}
p_{\theta}^N(y_{1:T}) & = & p_{\theta}^N(y_{1:t-2})  \widetilde{p}_{\theta}^N(y_{t+2:T})
\frac{1}{N}\sum_{l=1}^N\frac{1}{N^2}
\overrightarrow{W}_{t-1}(\overrightarrow{x}^{i(l)}_{t-2:t-1})\overleftarrow{W}_{t+1}(\overleftarrow{x}^{j(l)}_{t+1:t+2}) \times
\nonumber\\
& & 
\frac{f_{\theta}(x_t^l|\overrightarrow{x}_{t-1}^{i(l)})f_{\theta}(\overleftarrow{x}_{t+1}^{j(l)}|x_{t}^l)}{\xi_{t+1,\theta}(\overleftarrow{x}^j_{t+1})
\overrightarrow{\beta}_{t-1}^{i(l)}\overleftarrow{\beta}_{t+1}^{j(l)}
q_{t,\theta}(x_t^l|x_{t-1}^{i(l)},x_{t+1}^{j(l)})}g_{\theta}(y_t|x_t^l).
\label{eq:on_estimate}
\end{eqnarray}

\subsection{Unbiasedness and Central Limit Theorem}

We will give some analysis of the estimate \eqref{eq:on_estimate}; we denote this estimate
$p_{\theta}^N(y_{1:T})$.
We make an assumption (A1) which is detailed in the appendix. In addition, the notations for the expression of the asymptotic variance are also defined in the appendix.
For a function $\varphi:\mathbb{R}^d\rightarrow\mathbb{R}$ such that $\sup_{x\in\mathbb{R}^d}|\varphi(x)|<+\infty$, we write $\varphi\in\mathcal{B}_b(\mathbb{R}^d)$.
$\mathcal{N}_d(\mu,\Sigma)$ denotes a $d-$dimensional normal distribution with mean $\mu$ and covariance $\Sigma$; if $d=1$ the subscript $d$ is omitted.

\begin{theorem}\label{theo:clt}
We have
$$
\mathbb{E}[p_{\theta}^N(y_{1:T})] = p_{\theta}(y_{1:T}) \quad \forall \theta\in\Theta.
$$
In addition, assume (A1). Then for fixed $T> 2$, $t\in\{3,\dots,T-2\}$ and any $\theta\in\Theta$ we have that
$$
\sqrt{N}(p_{\theta}^N(y_{1:T})-p_{\theta}(y_{1:T})) \Rightarrow Z_{\theta}
$$
where $Z_{\theta}\sim\mathcal{N}(0,\sigma_{t,T}^2(\theta))$ with
\begin{eqnarray*}
\sigma_{t,T}^2(\theta) & = & 
\sigma^2_{\overrightarrow{\gamma}_{t-1,\theta}}\big(\overrightarrow{W}_{t-1}
\overleftarrow{\gamma}_{t+1,\theta}[\overleftarrow{W}_{t+1}^{\xi}
I_{gf}(.,\cdot)]\big) 
+ \sigma^2_{\overleftarrow{\gamma}_{t+1,\theta}}\big(
\overleftarrow{W}_{t+1}^{\xi}\overrightarrow{\gamma}_{t-1,\theta}(\overrightarrow{W}_{t-1}I_{gf}(\cdot,.))
\big)
\end{eqnarray*}
where, $\varphi\in\mathcal{B}_b(\mathbb{R}^{2d_x})$
\begin{eqnarray*}
\sigma^2_{\overrightarrow{\gamma}_{t-1,\theta}}(\varphi) & = & \sum_{q=1}^{t-1} \overrightarrow{\gamma}_{q,\theta}(1)^2
\overrightarrow{\eta}_{q,\theta}\bigg(\Big[\overrightarrow{Q}_{q,t-1}(\varphi)-
\overrightarrow{\eta}_{q,\theta}(\overrightarrow{Q}_{q,t-1}(\varphi))\Big]^2\bigg)\\
\sigma^2_{\overleftarrow{\gamma}_{t,\theta}}(\varphi) & = &  \sum_{q=0}^{T-t-1}\overleftarrow{\gamma}_{T-q,\theta}(1)^2 \overleftarrow{\eta}_{T-q,\theta}\bigg(
\Big[\overleftarrow{Q}_{T-q,t+1}(\varphi)-\overleftarrow{\eta}_{T-q,\theta}\overleftarrow{Q}_{T-q,t+1}(\varphi)\Big]^2
\bigg).
\end{eqnarray*}
\end{theorem}

\begin{rem}\label{meetingpointremark}
Under some additional mixing conditions, one may establish that the asymptotic variance $\sigma^2_{t,T}(\theta)$ when divided by $p_{\theta}(y_{1:T})^2$
(i.e.~the asymptotic variance associated to a normalized estimate) obeys the following inequality:
$
\sigma^2_{t,T}(\theta)/p_{\theta}(y_{1:T})^2 \leq C_1(\theta) (t-1) + C_2(\theta)(T-t)
$
where the first term is the error from the forward filter and the second from the backward filter. Unfortunately, this provides little intuition on how to select $t$ and it simply implies that if the forward algorithm works better, one should
choose $t$ large and vice versa.
\end{rem}

\begin{rem}\label{remark2}
Let $\varphi:\mathbb{R}^{d_x}\rightarrow \mathbb{R}$, $\varphi\in\mathcal{B}_b(\mathbb{R}^{d_x})$,  and consider $\mathbb{E}_{\theta}[\varphi(X_t)|y_{1:T}]$ where $3\leq t \leq T-2$ and the expectation is w.r.t.~the joint smoothing distribution, with density \eqref{eq:smoother}. Using the ideas in 
Section \ref{sec:est_marginal_like} one can show that an estimator of $\mathbb{E}_{\theta}[\varphi(X_t)|y_{1:T}]$ is
$$
\frac{1}{p_{\theta}^N(y_{1:T})}
p_{\theta}^N(y_{1:t-2})  \widetilde{p}_{\theta}^N(y_{t+2:T})
\frac{1}{N^3}\sum_{l=1}^N
\overrightarrow{W}_{t-1}(\overrightarrow{x}^{i(l)}_{t-2:t-1})\overleftarrow{W}_{t+1}(\overleftarrow{x}^{j(l)}_{t+1:t+2})
\frac{\varphi(x_t^l)f_{\theta}(x_t^l|\overrightarrow{x}_{t-1}^{i(l)})f_{\theta}(\overleftarrow{x}_{t+1}^{j(l)}|x_{t}^l)g_{\theta}(y_t|x_t^l)}{\xi_{t+1,\theta}(\overleftarrow{x}^j_{t+1})
\overrightarrow{\beta}_{t-1}^{i(l)}\overleftarrow{\beta}_{t+1}^{j(l)}
q_{t,\theta}(x_t^l|x_{t-1}^{i(l)},x_{t+1}^{j(l)})}.
$$
Denote the estimate as $p_{\theta,t}^N(\varphi)/p_{\theta}^N(y_{1:T})$ and set $\mathbb{E}_{\theta}[\varphi(X_t)|y_{1:T}] = p_{\theta,t}(\varphi)/p_{\theta}(y_{1:T})$. Standard calculations reveal (e.g.~\cite[pp.~301]{delmoral}) that
$$
\frac{p_{\theta,t}^N(\varphi)}{p_{\theta}^N(y_{1:T})} - 
\frac{p_{\theta,t}(\varphi)}{p_{\theta}(y_{1:T})} = \frac{p_{\theta}(y_{1:T})}{p_{\theta}^N(y_{1:T})}
p_{\theta,t}^N\bigg(
\frac{1}{p_{\theta}(y_{1:T})}\Big[\varphi - \frac{p_{\theta,t}(\varphi)}{p_{\theta}(y_{1:T})}\Big]
\bigg).
$$
Now, upon inspection of the proofs in the appendix, one can easily deduce:
\begin{itemize}
\item{$p_{\theta}(y_{1:T})/p_{\theta}^N(y_{1:T})$ will converge in probability to 1.}
\item{Let $\widetilde{\varphi} = 1/p_{\theta}(y_{1:T})[\varphi - p_{\theta,t}(\varphi)/p_{\theta}(y_{1:T})]$, then
$$ 
\sqrt{N} p_{\theta,t}^N\bigg(
\frac{1}{p_{\theta}(y_{1:T})}\Big[\varphi - \frac{p_{\theta,t}(\varphi)}{p_{\theta}(y_{1:T})}\Big]
\bigg) \Rightarrow Z_{\theta}(\widetilde{\varphi})
$$
where $Z_{\theta}(\widetilde{\varphi})\sim \mathcal{N}(0,\sigma^2_{t,T}(\widetilde{\varphi}))$,
$$
\sigma_{t,T}^2(\theta)  = 
\sigma^2_{\overrightarrow{\gamma}_{t-1,\theta}}\big(\overrightarrow{W}_{t-1}
\overleftarrow{\gamma}_{t+1,\theta}[\overleftarrow{W}_{t+1}^{\xi}
I_{gf\widetilde{\varphi}}(.,\cdot)]\big) 
+ \sigma^2_{\overleftarrow{\gamma}_{t+1,\theta}}\big(
\overleftarrow{W}_{t+1}^{\xi}\overrightarrow{\gamma}_{t-1,\theta}(\overrightarrow{W}_{t-1}I_{gf\widetilde{\varphi}}(\cdot,.))
\big)
$$
and  for $(\tilde{x}_{t-1},\tilde{x}_{t+1})\in\mathbb{R}^{2d_x}$
$$
I_{gf\widetilde{\varphi}}(\tilde{x}_{t-1},\tilde{x}_{t+1}) = \int_{\mathbb{R}^{d_x}} g_{\theta}(y_t|x_t)\widetilde{\varphi}(x_t) f_{\theta}(\tilde{x}_{t+1}|x_t)f_{\theta}(x_t|\tilde{x}_{t-1}) dx_t.
$$
See the appendix for further definitions of the notations.}
\end{itemize}
Hence, on using Slutsky's Lemma, one has a univariate CLT for an approximation of $\mathbb{E}_{\theta}[\varphi(X_t)|y_{1:T}]$. One can follow the ideas of \cite[pp.~301-302]{delmoral} to prove a multivariate CLT. It may be possible
to compare this estimate (through the asymptotic variance) relative to the one produced by the forward filtering backward smoothing algorithm; see \cite{dds1,douc}.
\end{rem}

\begin{rem}
The unbiased property allows one to use the SMC approximation of the generalized two-filter representation within a particle Markov chain Monte Carlo \cite{Andrieu_2010} algorithm. In \cite{persing}, we have established an appropriate target distribution in this context.
\end{rem}

\section{Numerical Examples}\label{sec:pmcmc}

\subsection{Measuring the New Estimate's Sensitivity to $t$}
Consider the linear Gaussian model provided in Section 4 of \cite{Fearnhead_2010}:
$X_0 \sim \mathcal{N}_2\left( \mu_0, \Sigma_0 \right)$,  $X_{n+1} \mid \left( X_{1:n}=x_{1:n}, Y_{1:n}=y_{1:n} \right) \sim \mathcal{N}_2\left( Fx_{n}, Q \right)$, 
$Y_{n} \mid \left( X_{1:n}=x_{1:n}, Y_{1:n-1}=y_{1:n-1} \right) \sim \mathcal{N}\left( Gx_{n}, R \right)$, with
$$
\begin{array}{c}
G =\left(1,0\right)\\[0.5cm]
F =\left( \begin{array}{cc}
1 & 1 \\
0 & 1 \end{array} \right)
\end{array} 
\begin{array}{c}
R =\tau^2\\[0.5cm]
Q =\nu^2 \left( \begin{array}{cc}
\frac{1}{3} & \frac{1}{2} \\
\frac{1}{2} & 1 \end{array} \right)
\end{array}\\
$$

We ran the two-filter SMC algorithm to calculate the marginal likelihood via \eqref{eq:on_estimate} for the instance where $T=300$.  Our objective is to observe how the choice of $t$ affects the accuracy and precision of the new estimate.  
A Kalman filter is used to allow us to choose $\xi_{n,\theta}(x_n) = \pi_{\theta}(x_n|y_{1:n-1})$ (the predictor); this corresponds to an extremely favourable choice (indeed one recovers the FFBS procedure, when considering the smoother).
In all simulations, we used the optimal importance distributions and $\beta$ resampling weights as in Appendix A of \cite{Fearnhead_2010}.  We set $N=300$.
We ran nine versions of the two-filter algorithm, with $t\in\{T/10,2T/10,\dots,9T/10\}$.  In each case, we plotted the variability of the estimate and compared \eqref{eq:on_estimate} to the maximum likelihood estimate provided by the Kalman filter. 
We ran many simulations for different pairs of values of the state noise, $\nu^2$, and the observation noise, $\tau^2$; specifically, we looked at $225$ possible pairings where $\nu^2$ and $\tau^2$ each ranged from $1$ to $98$.  
The results are displayed in Figure \ref{fig:simos1}.

We found the same phenomenon across all pairs of values of $\nu^2$ and $\tau^2$.  There is an increase in the variance of \eqref{eq:on_estimate} as $t$ approaches $T$ (i.e., when the new estimate relies more on the forward filter and less on the backward filter).  Furthermore, we see the accuracy of the estimate fall as $t$ approaches $T$.  These results are in accordance with the degeneracy measures of the two filters.  The forward filter's effective sample size (ESS) stays around $200$, while the backward filter's ESS \emph{never} drops below $N=300$ (due to the choice of
$\xi_{n,\theta}$ and the various proposals adopted).  The choice for $\xi_{n,\theta}$ ensures the backward filter's consistently strong performance.  This suggests that a better algorithm may result from removing any dependence on the forward filter and running a backward filter (with $\xi_{n,\theta}(x_n) = \pi_{\theta}(x_n|y_{1:n-1})$) from time $T$ to $1$; this point is discussed in Section \ref{sec:discussion}. Note that, in comparison to the standard SMC estimate, the results for the new estimate (for this model and the current settings) were superior w.r.t.~the variability of the estimate (results not shown).

\subsection{Comparing the Two-filter Decomposition to FFBSi}
Remark \ref{remark2} above parallels a similar result shown in \cite{douc} for another $\mathcal{O}(N)$ SMC smoothing approximation based on FFBS.  To explore this point further, we used the same example from \cite{Fearnhead_2010} to compare the two-filter SMC algorithm to the FFBSi algorithm in \cite{douc}.  We used both algorithms to calculate the expected value of the state of the hidden process given $y_{1:T=300}$ at time $t\in\{T/10,\dots, 9T/10\}$.  Both algorithms utilized $N=300$ particles. Note that FFBSi relies on rejection sampling, and so due to its stochastic running time, it is difficult to exactly match the computation times. Again, $50$ simulations per algorithm per $(\tau^2,\nu^2)$ pair for $225$ pairs are run; see Figure \ref{fig:simos2}.  We found that the two algorithms gave very similar results, although the two-filter decomposition yielded estimates of lower variance (see Figure \ref{fig:simos2}).  This is especially true at lower values of $t$, where, as above, the backward filter has more influence on the two-filter estimate.

\begin{figure}[h]
\begin{center}
\centering
\scalebox{0.44}{\includegraphics[trim = 25mm 80mm 10mm 80mm, clip]{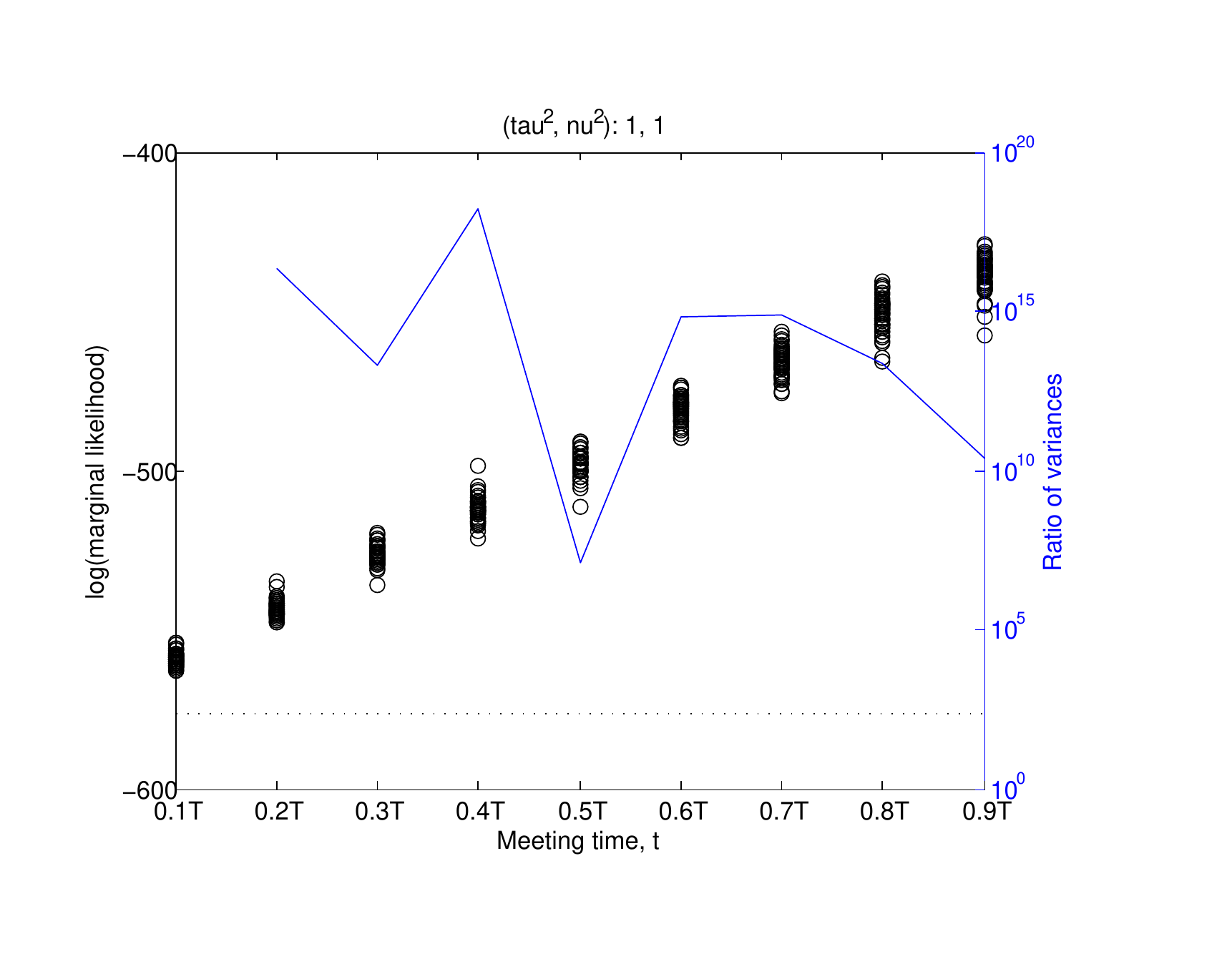}}
\scalebox{0.44}{\includegraphics[trim = 25mm 80mm 10mm 80mm, clip]{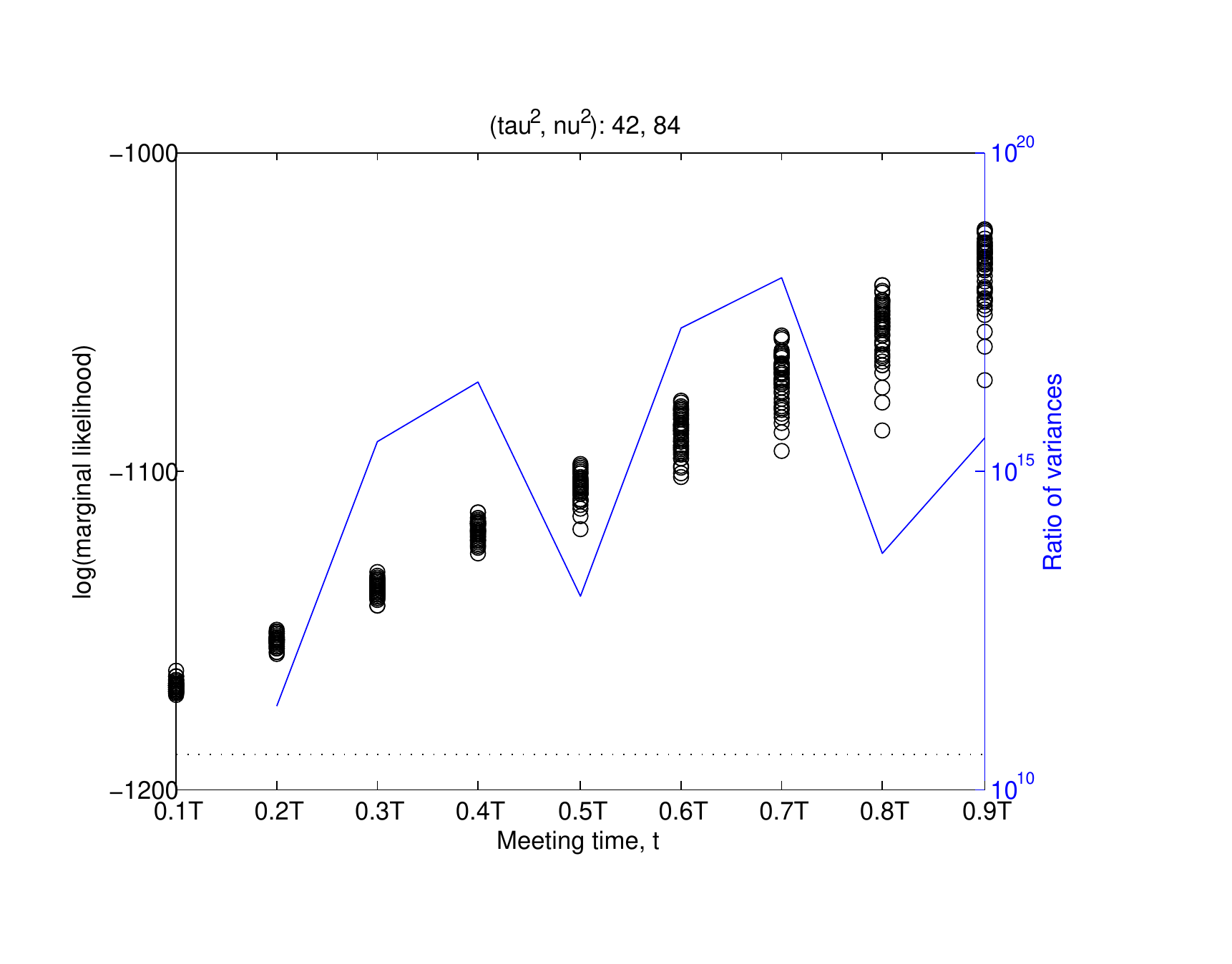}}
\scalebox{0.44}{\includegraphics[trim = 25mm 80mm 10mm 80mm, clip]{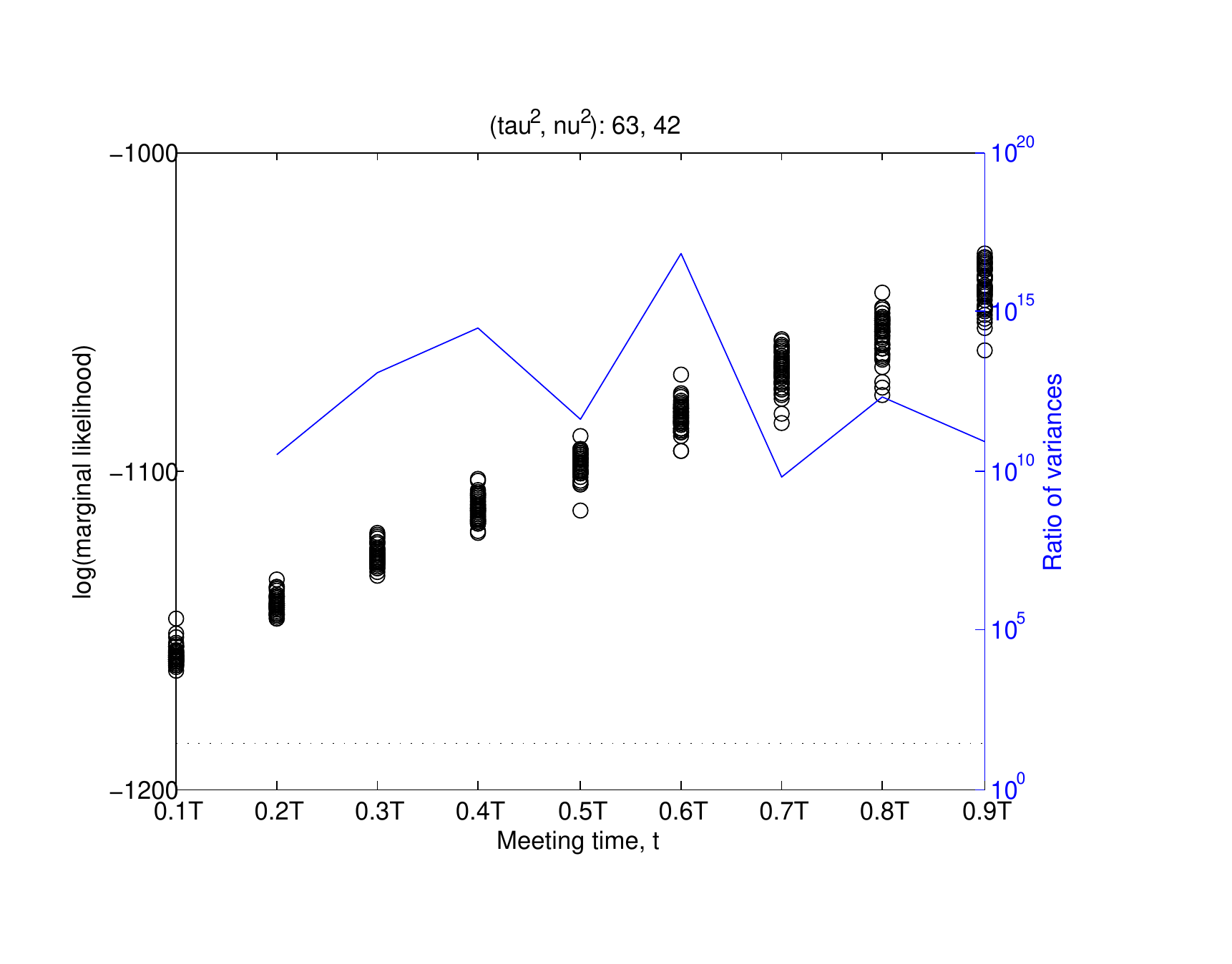}}
\scalebox{0.44}{\includegraphics[trim = 25mm 80mm 10mm 80mm, clip]{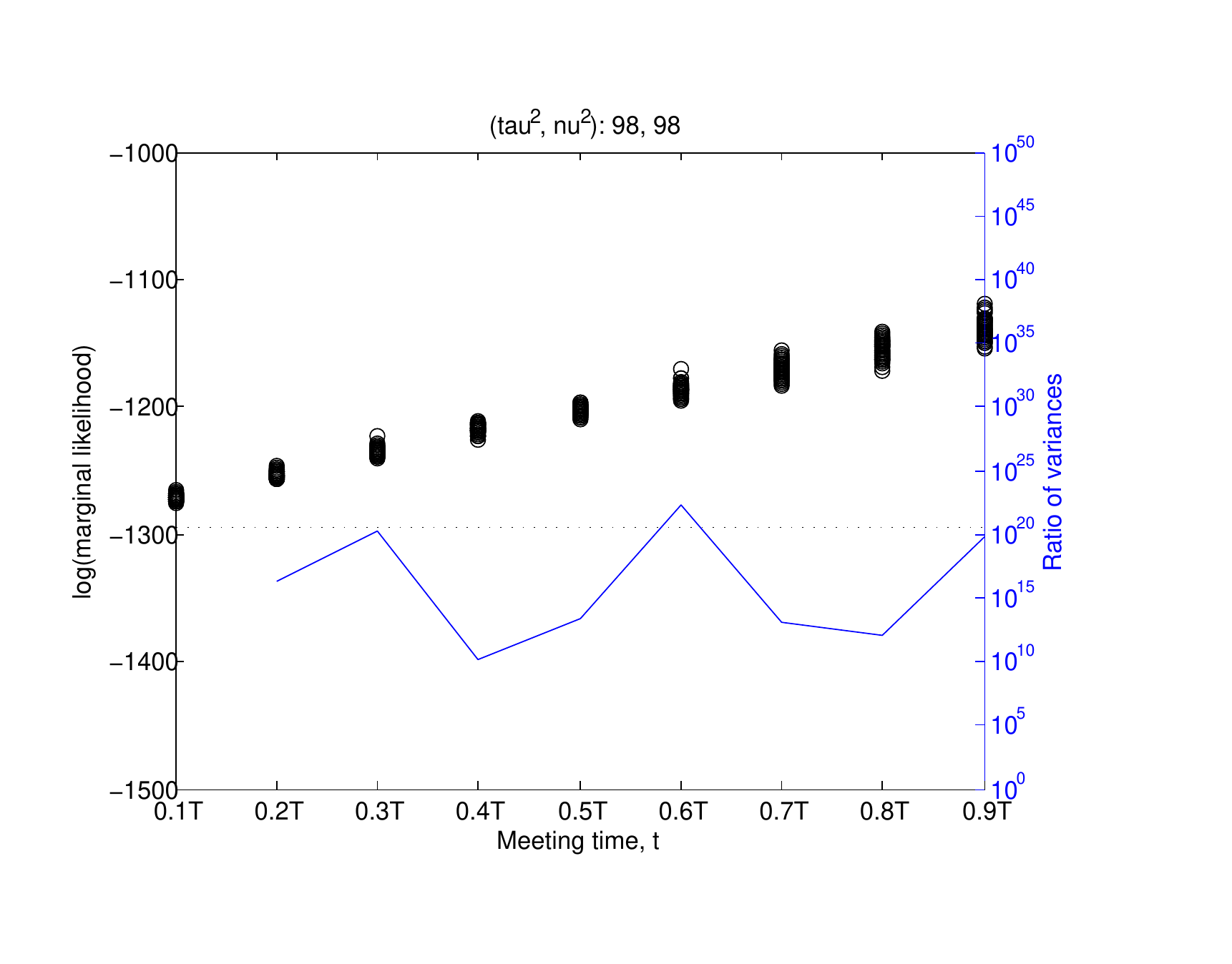}}
\caption{We present the output for some pairs of $\nu^2$ and $\tau^2$.  The circles, whose scale is on the left, give the $50$ simulated values of the logarithm of the marginal likelihood per time point.  The solid line, whose scale is on the right, measures the ratio of the variance of these $50$ values at each time point to the variance at the previous time point.  The dotted line gives the logarithm of the marginal likelihood as provided by the Kalman filter.}
\end{center}
\label{fig:simos1}
\end{figure}

\begin{figure}[h]
\begin{center}
\centering
\scalebox{0.44}{\includegraphics[trim = 25mm 80mm 10mm 80mm, clip]{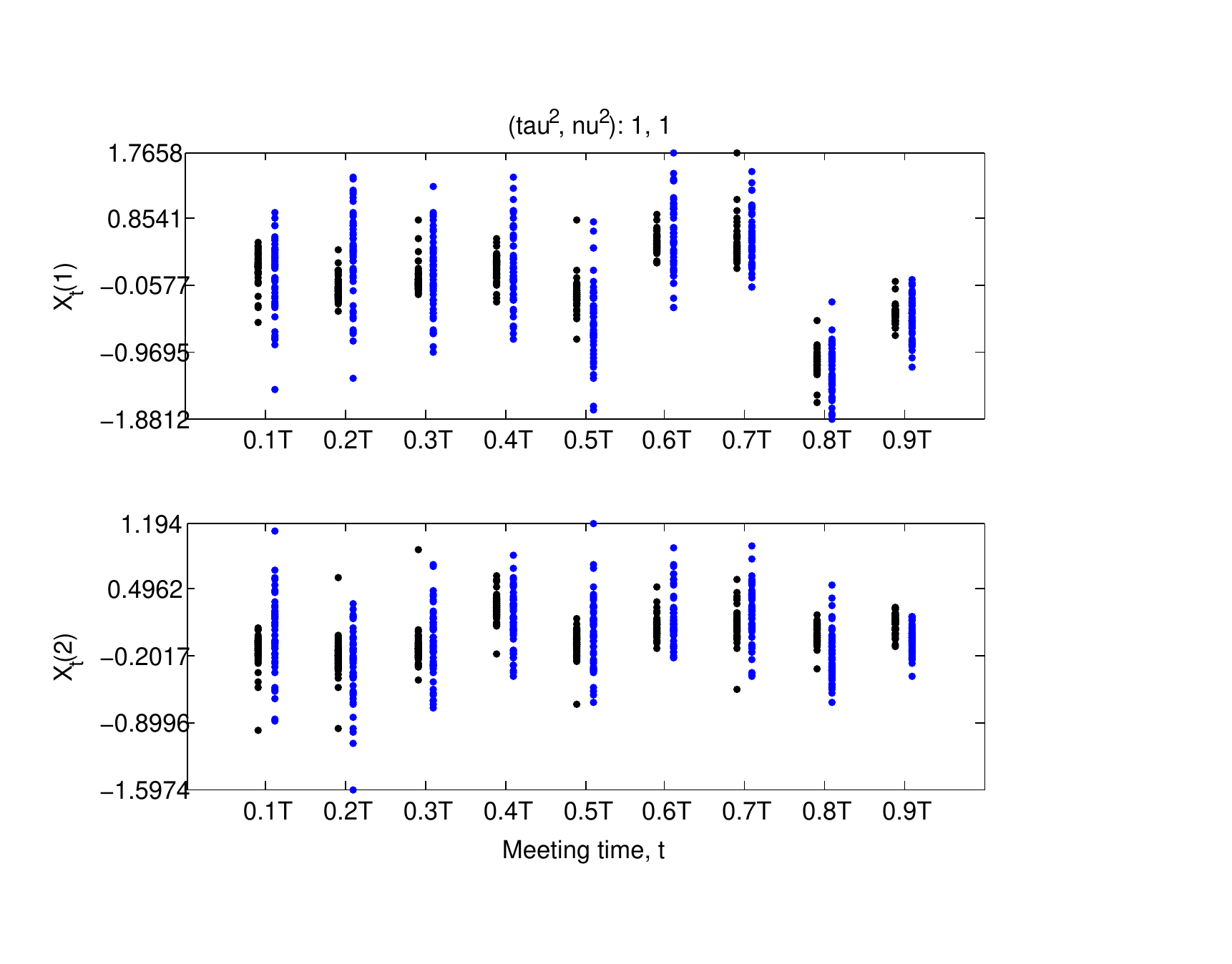}}
\scalebox{0.44}{\includegraphics[trim = 25mm 80mm 10mm 80mm, clip]{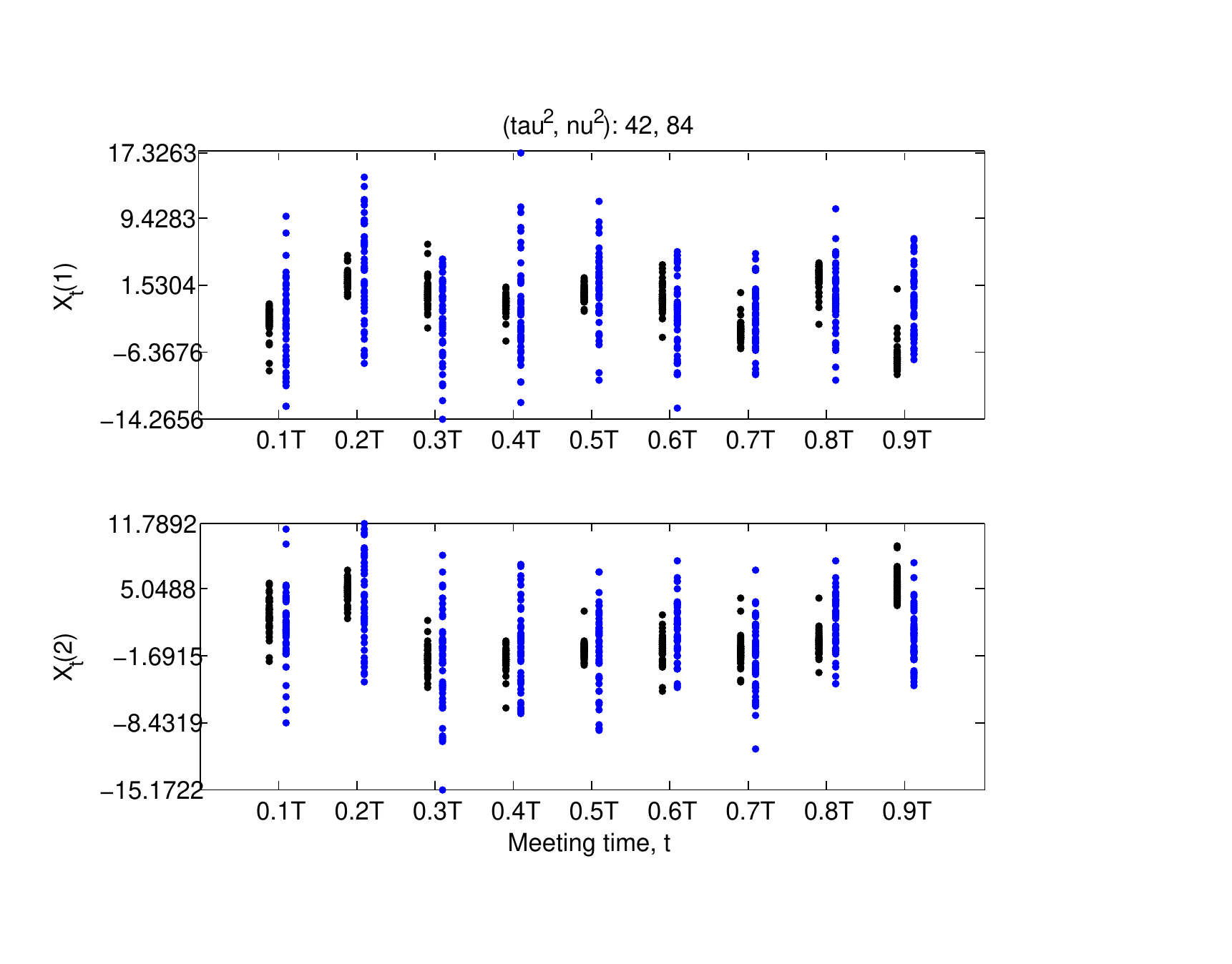}}
\scalebox{0.44}{\includegraphics[trim = 25mm 80mm 10mm 80mm, clip]{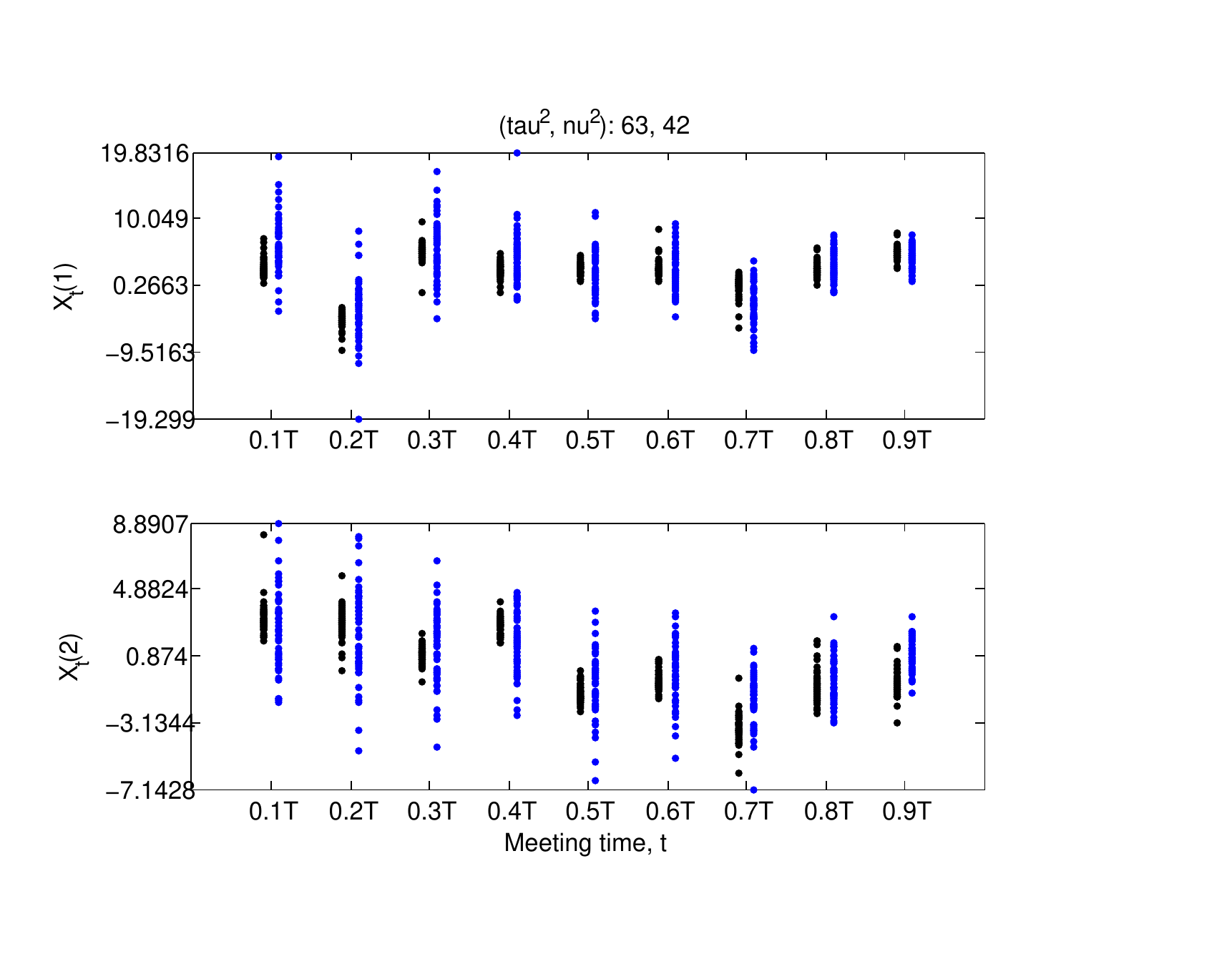}}
\scalebox{0.44}{\includegraphics[trim = 25mm 80mm 10mm 80mm, clip]{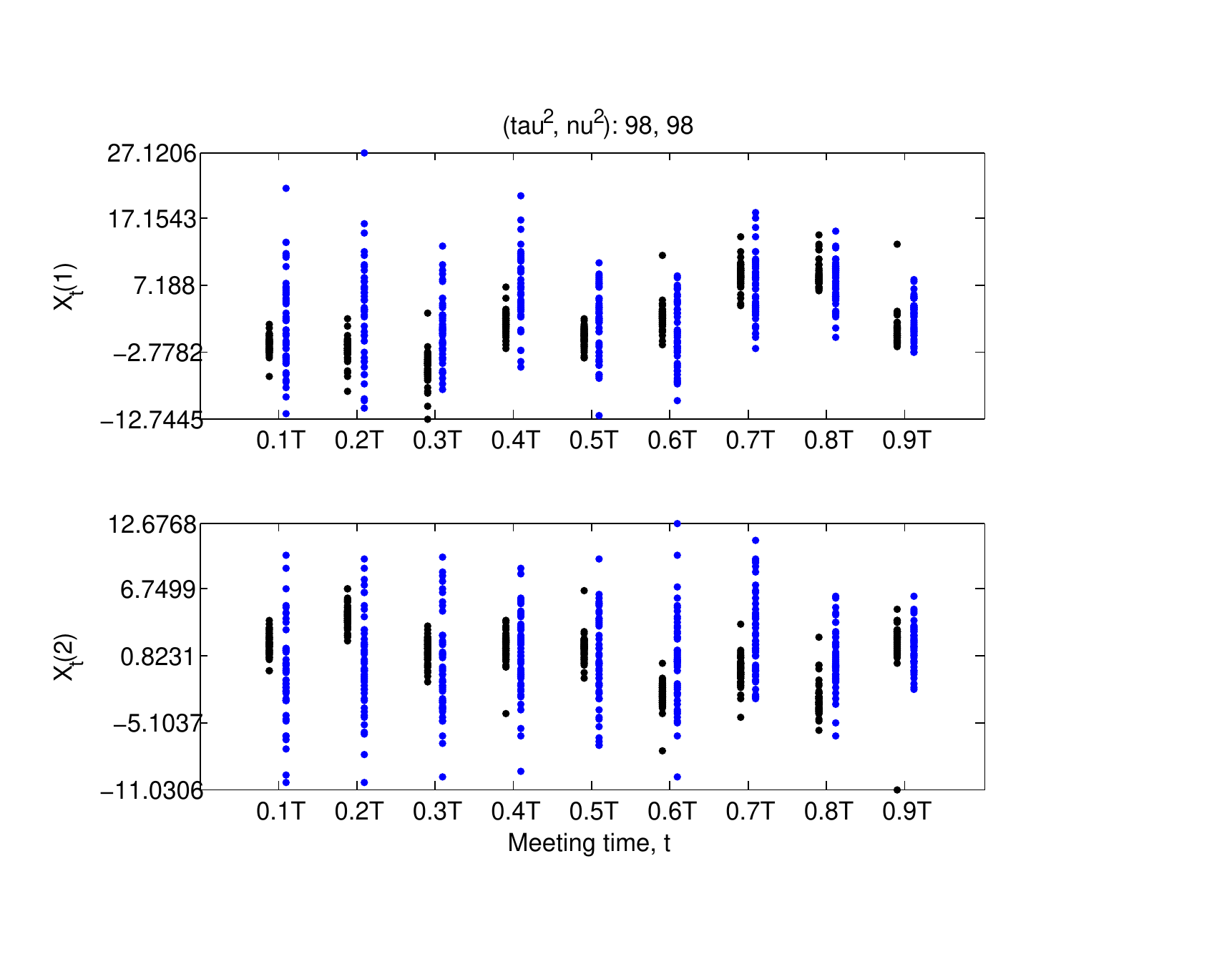}}
\caption{We present the output for some pairs of $\nu^2$ and $\tau^2$.  At each time point, the black dots (left) represent $50$ simulated expected values from the two-filter algorithm and the blue dots (right) represent $50$ estimates from FFBSi.  The first component of $\mathbb{E}[X_t|y_{1:T}]$ is on top, and the second component of $\mathbb{E}[X_t|y_{1:T}]$ is on the bottom.}
\end{center}
\label{fig:simos2}
\end{figure}

\section{Discussion}\label{sec:discussion}
In this note, we introduced a new $\mathcal{O}(N)$ estimate of the marginal likelihood using the generalized two-filter decomposition.  We established that this estimate is unbiased and proved a CLT, under some assumptions.  Numerical examples suggested that the new estimate is sensitive to changes in the meeting point of the forward and backward filters.  When choosing $\xi_{n,\theta}(x_n)=\pi_{\theta}(x_n|y_{1:T})$, the backward filter significantly outperforms the forward filter and it may be beneficial to remove the forward filter from the estimation procedure.  However, one can seldom make this choice for the $\xi_{n,\theta}$, and so we would like to approximate them.  In joint work with Prof.~A.~Doucet, we are exploring a smoothing algorithm where one introduces a discrete valued auxiliary variable $J\in \left\{1,...,N\right\} $ and considers the sequence of extended backward targets (where we condition upon the particles from a forward SMC algorithm)
$
\widetilde{\pi}_{\theta}(j, x_{n:T}|y_{1:T}) \propto f_\theta \left( \overrightarrow{x}_{n} \mid \overrightarrow{x}_{n-1}^j \right) g_{\theta}(y_n|x_n)
\bigg[\prod_{n=t+1}^T g_{\theta}(y_n|x_n)f_{\theta}(x_n|x_{n-1})\bigg] \quad n\in\{t,\dots,T\}.
$
The idea is to approximate the $\xi_{n,\theta}$ that are used above, via the forward filter.  

\subsubsection*{Acknowledgements}
This project has been initialized in joint research with Arnaud Doucet and we thank him for his input which has been critical.
We thank two referees and an associate editor, whose comments have greatly enhanced the article.

\appendix

\section{Proof of the CLT}

Here we describe a Feynman-Kac representation, which is used in the proof of the CLT. Let $t\in\{3,\dots,T-2\}$, with $T>2$ also fixed.

Define, the forward Feynman-Kac un-normalized $n-$time marginal, $n\in\{1,\dots,t-1\}$:
$$
\overrightarrow{\gamma}_{n,\theta}(dx_n) = \int \bigg[\prod_{p=1}^{n-1} \overrightarrow{W}_p(x_p)M_{p}(x_{p-1},dx_p)\bigg]M_{n}(x_{n-1},dx_n)
$$
with $x_p=(x_p',\tilde{x}_p)\in\mathbb{R}^{2d_x}$,
$M_1(x_{0},dx_0)=\delta_{x_0}(dx_1')q_{1,\theta}(\tilde{x}_1|x_1')d\tilde{x}_1$ and
\begin{eqnarray*}
M_{p}(x_{p-1},dx_p) & = & \delta_{\tilde{x}_{p-1}}(dx'_{p})q_{p,\theta}(\tilde{x}_p|x_{p}')d\tilde{x}_p.
\end{eqnarray*}
The normalized operator is 
$$
\overrightarrow{\eta}_{n,\theta}(dx_n) = \overrightarrow{\gamma}_{n,\theta}(dx_n)/\overrightarrow{\gamma}_{n,\theta}(1).
$$
We also define the forward semi-group operator:
$$
\overrightarrow{Q}_{p,n}(x_p,dx_n) = \int \prod_{q=p}^{n-1} \overrightarrow{W}_q(x_q) M_{q+1}(x_q,dx_{q+1})
$$
with $1\leq p \leq n\leq t-1$. 
 The selection mutation operator:
\begin{eqnarray*}
\overrightarrow{\Phi}_q(\overrightarrow{\eta}_{q-1,\theta})(\cdot) & = & \frac{\overrightarrow{\eta}_{q-1,\theta}(\overrightarrow{W}_{q-1}M_q(\cdot))}
{\overrightarrow{\eta}_{q-1,\theta}(\overrightarrow{W}_{q-1})} \quad q\in\{0,\dots,t-1\}
\end{eqnarray*}
with the conventions $\overrightarrow{\Phi}_1(\overrightarrow{\eta}_{0,\theta})=\overrightarrow{\eta}_{1,\theta}$.

Define, the  backward Feynman-Kac un-normalized $n-$time marginal, $n\in\{0,\dots,T-t-1\}$:
$$
\overleftarrow{\gamma}_{T-n,\theta}(dx_n) = \int \bigg[\prod_{p=0}^{T-n-1} \overleftarrow{W}_{T-p}(x_{T-p})M_{T-p}(x_{T-p+1},dx_{T-p})\bigg]M_{n}(x_{n+1},dx_n)
$$
with  $x_n=(x_n',\tilde{x}_n)\in\mathbb{R}^{2d_x}$,  $M_T(d\tilde{x}_T)  =  q_T(\tilde{x}_T)d\tilde{x}_T \delta_x(dx_T') $, $x\in\mathbb{R}^{d_x}$ an arbitrary point
\begin{eqnarray*}
M_n(x_{n+1},dx_n)  & = & q_{n,\theta}(\tilde{x}_n|x_n')dx_n\delta_{\tilde{x}_{n+1}}(dx_{n}') \quad n\in\{t+1,\dots,T-1\}.
\end{eqnarray*}
The normalized operator $\overleftarrow{\eta}_{T-n,\theta}=\overleftarrow{\gamma}_{T-n,\theta}(dx_n) /\overleftarrow{\gamma}_{T-n,\theta}(1) $.
Also define the semi-group operator
$$
\overleftarrow{Q}_{p,n}(x_p,dx_n) = \int \prod_{s=0}^{p-n-1} \overleftarrow{W}_{p-s}(x_{p-s}) M_{p-s-1}(x_{p-s},dx_{p-s-1})
$$
with $T\geq p \geq n \geq t+1$. Also
\begin{eqnarray*}
\overleftarrow{\Phi}_{T-q}(\overleftarrow{\eta}_{T-q+1,\theta})(\cdot) & = & \frac{\overleftarrow{\eta}_{T-q+1,\theta}(\overleftarrow{W}_{T-q+1}M_{T-q}(\cdot))}
{\overleftarrow{\eta}_{T-q+1,\theta}(\overleftarrow{W}_{T-q+1})} \quad q\in\{0,\dots,T-t-1\}
\end{eqnarray*}
and $\overleftarrow{\Phi}_{T}(\overleftarrow{\eta}_{T+1})=\overleftarrow{\eta}_{T}$.

We will use the notation
\begin{eqnarray*}
I_{gf}(\tilde{x}_{t-1},\tilde{x}_{t+1}) & = & \int_{\mathbb{R}^{d_x}} g_{\theta}(y_t|x_t) f_{\theta}(\tilde{x}_{t+1}|x_t)f_{\theta}(x_t|\tilde{x}_{t-1}) dx_t\\
W_{t+1}^{\xi}(x_{t+1}) & = & \frac{\overleftarrow{W}_{t+1}(x_{t+1})}{\xi_{t+1,\theta}(\tilde{x}_{t+1})}
\end{eqnarray*}
with
\begin{eqnarray*}
\mu_{t-1}(\overrightarrow{W}_{t-1}
\overleftarrow{\gamma}_{t+1,\theta}[\overleftarrow{W}_{t+1}^{\xi}
I_{gf}(.,\cdot)]) & = & \int \mu_{t-1}(dx_{t-1})\overrightarrow{W}_{t-1}(x_{t-1})[\int \overleftarrow{\gamma}_{t+1,\theta}(dx_{t+1}) 
\overleftarrow{W}_{t+1}^{\xi}(x_{t+1})
I_{gf}(\tilde{x}_{t-1},\tilde{x}_{t+1})]\\
\mu_{t+1}(\overleftarrow{W}_{t+1}^{\xi}\overrightarrow{\gamma}_{t-1,\theta}(\overrightarrow{W}_{t-1}I_{gf}(\cdot,.))) & = &
\int \mu_{t+1}(dx_{t+1}) \overleftarrow{W}_{t+1}^{\xi}(x_{t+1}) [\int 
\overrightarrow{\gamma}_{t-1,\theta}(dx_{t-1})\overrightarrow{W}_{t-1}(x_{t-1})I_{gf}(\tilde{x}_{t-1},\tilde{x}_{t+1})]
\end{eqnarray*}
for $\sigma-$finite measures $\mu_{t-1},\mu_{t+1}$.

Using the above notations, we can write
$$
p_{\theta}^N(y_{1:T}) =\overrightarrow{\gamma}_{t-1}^N(1)\overleftarrow{\gamma}_{t+2}^N(1)\frac{1}{N}\sum_{l=1}^N 
\frac{\overrightarrow{W}_{t-1}(\overrightarrow{x}^{i(l)}_{t-1})\overleftarrow{W}_{t+1}(\overleftarrow{x}^{j(l)}_{t+1}) 
f_{\theta}(x_t^l|\overrightarrow{\tilde{x}}_{t-1}^{i(l)})f_{\theta}(\overleftarrow{\tilde{x}}_{t+1}^{j(l)}|x_{t}^l)}{N^2 \xi_{t+1,\theta}(\overleftarrow{\tilde{x}}^j_{t+1})
\overrightarrow{\beta}_{t-1}^{i(l)}\overleftarrow{\beta}_{t+1}^{j(l)}
q_{t,\theta}(x_t^l|\overrightarrow{\tilde{x}}_{t-1}^{i(l)},\overleftarrow{\tilde{x}}_{t+1}^{j(l)})}g_{\theta}(y_t|x_t^l)
$$
with
\begin{eqnarray*}
\overrightarrow{\gamma}_{t-1}^N(1)& = & \prod_{p=1}^{t-2} \frac{1}{N}\sum_{i=1}^N \overrightarrow{W}_p(\overrightarrow{x}_p^i)\\
\overleftarrow{\gamma}_{t+2}^N(1)& = & \prod_{p=0}^{T-t-2} \frac{1}{N}\sum_{i=1}^N \overrightarrow{W}_{T-p}(\overrightarrow{x}_{T-p}^i).
\end{eqnarray*}

To prove the central limit theorem (CLT), we make use of the following assumption, which is similar to $(H)_m$ ($m=2$) of \cite{cerou1}. It is used to control remainder terms, when constructing a CLT. It implies that the backward Markov proposal kernels mix very quickly.
\begin{hypA} 
\label{hyp:A}
\begin{enumerate}
\item{The incremental weights all satisfy:
$$
\forall 1\leq n \leq t-1 \quad \delta_{\theta} = \sup_{x,y}\frac{\overrightarrow{W}_n(x)}{\overrightarrow{W}_n(y)} <\infty \quad
\forall t+1\leq n \leq T \quad \delta_{\theta} = \sup_{x,y}\frac{\overleftarrow{W}_n(x)}{\overleftarrow{W}_n(y)} <\infty
$$
For each $\theta\in\Theta$ there exist $0<\underline{C}_{\theta}<\overline{C}_{\theta}<\infty$ such that for every $x,x'\in\mathbb{R}^{d_x}$, $n\in\{1,\dots,T\}$, $y_n\in\mathbb{R}^{d_y}$
$$
\underline{C}_{\theta}\leq f_{\theta}(x'|x) \leq \overline{C}_{\theta} \quad  \underline{C}_{\theta}\leq \xi_{n,\theta}(x) \leq \overline{C}_{\theta} \quad \underline{C}_{\theta} \leq g_{\theta}(y_n|x) \leq \overline{C}_{\theta}.
$$
In addition, for each $\theta\in\Theta$, there exist $0<\underline{C}_{\theta}<\overline{C}_{\theta}<\infty$ as above, such that for each $x_t,x_{t-1},x_t\in\mathbb{R}^{d_x}$, $i\in\{1,\dots,N\}$
$$
\underline{C}_{\theta}\leq q_{t,\theta}(x_t|x_{t-1},x_{t+1}) \leq \overline{C}_{\theta} \quad \underline{C}_{\theta} \leq \overrightarrow{\beta}_{t-1}^i \leq  \overline{C}_{\theta}
\quad
\underline{C}_{\theta} \leq \overleftarrow{\beta}_{t+1}^i \leq  \overline{C}_{\theta}.
$$
}
\item{For $m=2$ and some sequence of numbers $\omega_p^{(m)}\in[1,\infty)$ such that for each $p\in\{-1,\dots,T-t-m\}$ and any $(x,x')\in \mathbb{R}^{2d_x}$ we have
$$
M_{T-p,T-p-m}(x,dy) \leq \omega_p^{(m)} M_{T-p,T-p-m}(x',dy)
$$
where $M_{p,q} = M_{p-1}\dots M_{q}$, $p\geq q$.}
\end{enumerate}
\end{hypA}

\begin{proof}[Proof of Theorem \ref{theo:clt}]
We have that:
$$
\mathbb{E}[p_{\theta}^N(y_{1:T})|\overrightarrow{\mathscr{F}}_{t-1}^N\otimes\overleftarrow{\mathscr{F}}_{t+1}^N] = \overrightarrow{\gamma}_{t-1}^N\otimes\overleftarrow{\gamma}_{t+1}^N(\overrightarrow{W}_{t-1}\overleftarrow{W}_{t+1}^{\xi}I_{gf})
$$
where $\overrightarrow{\mathscr{F}}_{t-1}^N$ and $\overleftarrow{\mathscr{F}}_{t+1}^N$ are the filtrations
generated by the forward and backward particle systems up-to time $t-1$ and $t+1$ respectively.
We can use the decomposition of \cite{delmoral} to obtain the following formula:
$$
\mathbb{E}[p_{\theta}^N(y_{1:T})|\overrightarrow{\mathscr{F}}_{t-1}^N\otimes\overleftarrow{\mathscr{F}}_{t+1}^N] - p_{\theta}(y_{1:T}) = \alpha(N) + \beta(N) + R(N)
$$
where
\begin{eqnarray*}
\alpha(N) & = & \sum_{q=1}^{t-1} \overrightarrow{\gamma}_{q}^N(1)[\overrightarrow{\eta}_q^N-\overrightarrow{\Phi}_{q}(\overrightarrow{\eta}_{q-1}^N)]
(\overrightarrow{Q}_{q,t-1}[\overrightarrow{W}_{t-1}\overleftarrow{\gamma}_{t+1}(\overleftarrow{W}_{t+1}^{\xi}I_{gf}(.,\cdot))])\\
\beta(N) & = &
\sum_{q=0}^{T-t-1} \overleftarrow{\gamma}_{T-q}^N(1)[\overleftarrow{\eta}_{T-q}^N-\overleftarrow{\Phi}_{T-q}(\overleftarrow{\eta}_{T-q-1}^N)]
(\overleftarrow{Q}_{T-q,t}[W_{t+1}^{\xi}\overrightarrow{\gamma}_{t-1}(W_{t-1}I_{gf}(\cdot,.))])\\
R(N) & = & \sum_{q=1}^{t-1} \overrightarrow{\gamma}_{q}^N(1)[\overrightarrow{\eta}_q^N-\overrightarrow{\Phi}_{q}(\overrightarrow{\eta}_{q-1}^N)]
(\overrightarrow{Q}_{q,t-1}[\overrightarrow{W}_{t-1}[\overleftarrow{\gamma}_{t+1}^N-\overleftarrow{\gamma}_{t+1,\theta}](\overleftarrow{W}_{t+1}^{\xi}I_{gf}(.,\cdot))]).
\end{eqnarray*}
It is straightforward to verify that the expectation of this quantity is exactly zero, which establishes the unbiased property. 

By using the  Marcinicwiez-Zygmund inequality
$$
\mathbb{E}[|\sqrt{N}(p_{\theta}^N(y_{1:T})-\mathbb{E}[p_{\theta}^N(y_{1:T})|\overrightarrow{\mathscr{F}}_{t-1}^N\otimes\overleftarrow{\mathscr{F}}_{t+1}^N]) |]
\leq \frac{C_{\theta}}{N^2}\mathbb{E}[|\overrightarrow{\gamma}_{t-1}^N(1)\overleftarrow{\gamma}_{t+2}^N(1)|]
$$
for some $C_{\theta}<+\infty$. For any fixed $t,T$, $\sup_{N\geq 1}\mathbb{E}[\overrightarrow{\gamma}_{t-1}^N(1)^2]^{1/2}<\infty$ and $\sup_{N\geq 1}\mathbb{E}[\overleftarrow{\gamma}_{t+2}^N(1)^2]^{1/2}<\infty$ (see the proof of Lemma \ref{lem:techlem}), thus, via Cauchy-Schwarz, we can deduce that (note that $\rightarrow_{\mathbb{P}}$ denotes convergence in probability)
$$
\sqrt{N}(p_{\theta}^N(y_{1:T})-\mathbb{E}[p_{\theta}^N(y_{1:T})|\overrightarrow{\mathscr{F}}_{t-1}^N\otimes\overleftarrow{\mathscr{F}}_{t+1}^N]) \rightarrow_{\mathbb{P}} 0.
$$

The weak convergence of $\sqrt{N}\alpha(N)$ and $\sqrt{N}\beta(N)$ can be obtained by the independence of the terms and \cite[Proposition 9.4.1]{delmoral}. By Lemma \ref{lem:techlem} the remainder $\sqrt{N}R(N)$ converges to zero in probability and we can conclude
the result.
\end{proof}

\begin{lem}\label{lem:techlem}
Assume (A1). Then for fixed $T> 2$, $t\in\{3,\dots,T-2\}$ we have that
$$
\sqrt{N}R(N)  =  \sqrt{N}\sum_{q=1}^{t-1} \overrightarrow{\gamma}_{q}^N(1)[\overrightarrow{\eta}_q^N-\overrightarrow{\Phi}_{q}(\overrightarrow{\eta}_{q-1}^N)]
(\overrightarrow{Q}_{q,t-1}[\overrightarrow{W}_{t-1}[\overleftarrow{\gamma}_{t+1}^N-\overleftarrow{\gamma}_{t+1,\theta}](\overleftarrow{W}_{t+1}^{\xi}I_{gf}(.,\cdot))])
\rightarrow_{\mathbb{P}} 0.
$$
\end{lem}
\begin{proof}
To shorten the subsequent notations, define:
\begin{eqnarray*}
\xi_{q,t-1}^N(x) & = & \overrightarrow{Q}_{q,t-1}[\overrightarrow{W}_{t-1}[\overleftarrow{\gamma}_{t+1}^N-\overleftarrow{\gamma}_{t+1,\theta}](\overleftarrow{W}_{t+1}^{\xi}I_{gf}(.,\cdot)](x)\\
\bar{\xi}_{q,t-1}^N & = & \sup_x \overrightarrow{W}_{t-1}(x) \sup_x \overrightarrow{Q}_{q,t-1}(|[\overleftarrow{\gamma}_{t+1}^N-\overleftarrow{\gamma}_{t+1,\theta}](\overleftarrow{W}_{t+1}^{\xi}I_{gf}(.,\cdot)|)(x).
\end{eqnarray*}
It is remarked that for any bounded function $\varphi$, $\sup_x \overrightarrow{Q}_{q,t-1}(\varphi)(x)<\infty$ by assumption.

We will now show that $\sqrt{N}R(N)$ will go-to zero in $\mathbb{L}_1$. To that end, we can consider the expectation of each summand in the series for $R(N)$.
We have
$$
\mathbb{E}\Big[\Big|\overrightarrow{\gamma}_{q}^N(1)[\overrightarrow{\eta}_q^N-\overrightarrow{\Phi}_{q}(\overrightarrow{\eta}_{q-1}^N)]\Big(\frac{\xi_{q,t-1}^N}{\bar{\xi}_{q,t-1}^N}\Big) \bar{\xi}_{q,t-1}^N\Big|\Big]
\leq \mathbb{E}\Big[\Big|\overrightarrow{\gamma}_{q}^N(1)[\overrightarrow{\eta}_q^N-\overrightarrow{\Phi}_{q}(\overrightarrow{\eta}_{q-1}^N)]\Big(\frac{\xi_{q,t-1}^N}{\bar{\xi}_{q,t-1}^N}\Big)\Big|^2\Big]^{1/2}
\mathbb{E}[(\bar{\xi}_{q,t-1}^N)^2]^{1/2}
$$
where we have used Cauchy-Schwarz. For the first expectation on the R.H.S.~one can condition on $\overrightarrow{\mathscr{F}}_{q-1}^N\otimes \overleftarrow{\mathscr{F}}_{t+1}^N$ and apply the Marcinicwiez-Zygmund inequality
(noting that $\sup_x|\xi_{q,t-1}^N(x)|/\bar{\xi}_{q,t-1}^N$ is upper-bounded by a finite deterministic constant) to obtain that
$$
\mathbb{E}\Big[\Big|\overrightarrow{\gamma}_{q}^N(1)[\overrightarrow{\eta}_q^N-\overrightarrow{\Phi}_{q}(\overrightarrow{\eta}_{q-1}^N)]\Big(\frac{\xi_{q,t-1}^N}{\bar{\xi}_{q,t-1}^N}\Big)\Big|^2\Big]^{1/2}
\leq \frac{C}{\sqrt{N}}\mathbb{E}[\overrightarrow{\gamma}_{q}^N(1)^2]^{1/2}.
$$
Note that for each $q$, $\mathbb{E}[\overrightarrow{\gamma}_{q}^N(1)^2]^{1/2}<\infty$ (e.g.~\cite[Corollary 5.2]{cerou1}, or by using the upper-bound on the $\overrightarrow{W}_{n}$).

Now, we move onto the expression $\mathbb{E}[(\bar{\xi}_{q,t-1}^N)^2]^{1/2}$. From the definition of $\bar{\xi}_{q,t-1}^N$, we have that
$$
\mathbb{E}[(\bar{\xi}_{q,t-1}^N)^2]^{1/2} =  \sup_x \overrightarrow{W}_{t-1}(x)  \sup_x \overrightarrow{Q}_{q,t-1}(1)(x) \mathbb{E}[\overrightarrow{\overline{Q}}_{q,t-1}(|[\overleftarrow{\gamma}_{t+1}^N-\overleftarrow{\gamma}_{t+1,\theta}](\overleftarrow{W}_t^{\xi}I_{gf}(.,\cdot)|)^2]^{1/2}
$$
where $\overrightarrow{\overline{Q}}_{q,t-1}(\cdot)(x):=\sup_x\overrightarrow{Q}_{q,t-1}(\cdot)(x)/\sup_x\overrightarrow{Q}_{q,t-1}(\cdot)(x)$. Application of Jensen's inequality and Fubini leads to
$$
\mathbb{E}[(\bar{\xi}_{q,t-1}^N)^2]^{1/2} \leq \sup_x \overrightarrow{W}_{t-1}(x)  \sup_x \overrightarrow{Q}_{q,t-1}(1)(x)
\overrightarrow{\overline{Q}}_{q,t-1}\Big( \mathbb{E}[|[\overleftarrow{\gamma}_{t+1}^N-\overleftarrow{\gamma}_{t+1,\theta}](\overleftarrow{W}_t^{\xi}I_{gf}(.,\cdot))|^2]\Big)^{1/2}.
$$
Then by \cite[Theorem 5.1, Corollary 5.2]{cerou1} (it is remarked that the corollary of that paper can be adapted to deal when $\overleftarrow{\gamma}_{t+1}$ integrates a bounded function), it follows for $N$ large enough relative to $T-t$ (we will take $N$ to infinity and $T-t$ is fixed) there exist some finite constant $C(T,t)$ that depends upon $T,t$ but not $q$ or $x_{t-1}$ such that
$$
\mathbb{E}[(\bar{\xi}_{q,t-1}^N)^2]^{1/2} \leq \sup_x G_{t-1}(x)  \sup_x Q_{q,t-1}(1)(x)\frac{C(T,t)}{\sqrt{N}}.
$$
Hence we have that:
$$
\sqrt{N}\mathbb{E}[|R(N)|] \leq \frac{C(T,t)}{\sqrt{N}}
$$
where $C(T,t)$ is some finite constant that may grow with $T$. We thus conclude as $T<\infty$.
\end{proof}



\end{document}